%% file: sppr.tex
\RequirePackage{amsmath}
\documentclass[runningheads]{llncs}
\usepackage{amsmath,amssymb,amsfonts,paralist,graphicx,color}
\usepackage[colorlinks=true,linkcolor=blue]{hyperref}

\usepackage[normalem]{ulem}
\usepackage[utf8]{inputenc}

\usepackage{amsthm}

\usepackage{mathtools} 
\usepackage{wrapfig} 
\usepackage{nicefrac}
\usepackage{xspace}

\usepackage[vlined,ruled,noend]{algorithm2e}

\usepackage{caption} 
\usepackage{subcaption}
\usepackage{enumitem}
\usepackage{sidecap} 
\usepackage{xstring}

\usepackage{tikz}
\tikzset{
  box/.style={rectangle,draw=black, thick},
}
\usetikzlibrary{calc}

\newcommand*{\GetListMember}[2]{%
    \edef\dotheloop{%
    \noexpand\foreach \noexpand\a [count=\noexpand\i] in {#1} {%
        \noexpand\IfEq{\noexpand\i}{#2}{\noexpand\a\noexpand\breakforeach}{}%
    }}%
    \dotheloop
    \par%
}%

\newcommand{\clauseGadget}{
	\pgfmathsetmacro\radiusA{1.2}
	\pgfmathsetmacro\radiusB{1.8}
	\node[vertex] (c) at (0, 0) {};
    \node[vertex] (v{0}) at ($(c) + (150:\radiusA)$) {};
	\foreach \j in {1, ..., 8} {
    	\node[vertex] (v{\j}) at ($(c) + (150 - \j * 40:\radiusA)$) {};
       	\draw (c) to (v{\j});
		\draw (v{\number\numexpr \j - 1 \relax}) to (v{\j});
	}
    \draw (c) to (v{0});
	\draw (v{0}) to (v{8});
	\foreach \j in {0, 1, 2} {
    	\node[vertex] (v{9+\j}) at ($(c) + (150 - \j * 120:\radiusB)$) {};
       	\draw (v{\number\numexpr 0+\j*3 \relax}) to (v{9+\j});
	}
}

\newcommand{\clauseGadgetDrawLabels}{
        \pgfmathtruncatemacro\k{1+0}
 	\node at ($(c) + (150 - 0 * 120:\radiusB+.4)$) {$\clauseVertex_{j, 1}$};
        \pgfmathtruncatemacro\k{1+1}
 	\node at ($(c) + (150 - 1 * 120:\radiusB+.4)$) {$\clauseVertex_{j, 2}$};
        \pgfmathtruncatemacro\k{1+2}
 	\node at ($(c) + (160 - 2 * 120:\radiusB+.2)$) {$\clauseVertex_{j, 3}$};
}

\newcommand{\clauseRounding}[2][1.7]{
	\coordinate (c) at (0, 0);
	\foreach \j in {1, ..., 9} 
		\node (v{\j}) at ($(c) + (-70 - \j * 40:#1)$) {\GetListMember{#2}{\j}};
}

\newcommand{\XSays}[3]{{\color{#2}
      {$\rule[-0.12cm]{0.2in}{0.5cm}$\fbox{\tt
            #1:} }%
      \itshape #3
      \marginpar{\color{#2}\tt #1}%
      \def\comment{#3}\def\empty{}\ifx\comment\empty\else
      {$\rule[0.1cm]{0.3in}{0.1cm}$\fbox{\tt
            end}$\rule[0.1cm]{0.3in}{0.1cm}$} \fi
   }%
}
\newcommand{\HH}[1]{}
\newcommand{\DK}[1]{}

\newcommand{\abs}[1]{\ensuremath{\left| \, #1 \,\right|}\xspace}

\newcommand{\Wlog}{W.\,l.\,o.\,g.\xspace}

\newcommand{\eps}{\ensuremath\varepsilon\xspace}

\title{Shortest-Path-Preserving Rounding}

\author{Herman Haverkort \and David K\"ubel \and Elmar Langetepe}
\authorrunning{H. Haverkort et al.}
\institute{Universit\"at Bonn, Germany, \email{\{haverkort,dkuebel,elmar.langetepe\}@uni-bonn.de}}

\begin{document}
\maketitle              
\input{includes/abstract.tex}

\input{includes/introduction.tex}

\input{includes/complexity.tex}
\input{includes/algorithms.tex}

\input{includes/conclusion.tex}

\input{includes/bibliography.tex}

\end{document}

%% file: includes/abstract.tex
\begin{abstract}
Various applications of graphs, in particular applications related to finding shortest paths, naturally get inputs with real weights on the edges. However, for algorithmic or visualization reasons, inputs with integer weights would often be preferable or even required. This raises the following question: given an undirected graph with non-negative real weights on the edges and an error threshold~$\eps$, how efficiently can we decide whether we can round all weights such that shortest paths are maintained, and the change of weight of each shortest path is less than~$\eps$? So far, only for path-shaped graphs a polynomial-time algorithm was known. In this paper we prove, by reduction from 3-SAT, that, in general, the problem is NP-hard. However, if the graph is a tree with $n$ vertices, the problem can be solved in $O(n^2)$ time.
\end{abstract}

%% file: includes/introduction.tex
\newcommand{\graph}{\ensuremath{G}\xspace}
\newcommand{\vertexSet}{\ensuremath{V}\xspace}
\newcommand{\edgeSet}{\ensuremath{E}\xspace}
\newcommand{\weights}{\ensuremath{\omega}\xspace}
\newcommand{\roundedWeights}{\ensuremath{\tilde{\omega}}\xspace}
\newcommand{\roundedGraph}{\ensuremath{\tilde{\graph}}\xspace}
\newcommand{\SPPR}{\textsc{WSPPR}\xspace}
\newcommand{\SSPPR}{\textsc{SSPPR}\xspace}
\newcommand{\DPR}{\textsc{DPR}\xspace}
\newcommand{\shortestPath}[2][]{\ensuremath{\pi_{#1}(#2)}\xspace}
\newcommand{\rounding}[1][\eps]{\ensuremath{#1}-rounding\xspace}
\newcommand{\roundings}[1][\eps]{\ensuremath{#1}-roundings\xspace}
\section{Introduction}
Consider a transportation network, modelled as an undirected graph, with a weight function on the edges that represents the time (or cost) it takes to travel each edge. In this paper, we will also refer to the weights as lengths. For several applications, it would be advantageous if the weights are small integers. For example, one could then draw a zone map of the network such that the number of zone boundaries crossed by each shortest path corresponds to the weight of the path~\cite{Haverkort2014}. This raises the following question: given a transportation network with weights for all edges, normalized such that weight~1 corresponds to the intended zone diameter of the map, how can we round the weights to integers such that shortest paths are maintained? If we can do this, for a well-chosen zone diameter, then we can draw a zone map that provides a fairly accurate representation of travel costs, and which is easier to read and use than a map in which the true travel costs are written in full detail next to each edge\footnote{The zone diameter should be chosen for a good balance between precision and complexity of the map. For an unambiguous map drawing it may also be required that after rounding, there are no cycles of weight 1.}.\HH{Is the footnote useful, or just distracting?} Other applications that could take advantage of rounded weights include algorithms to compute shortest paths: there are algorithms that are more efficient with small integer weights than with arbitrary, real weights~\cite{Thorup2003}.
Funke and Storandt~\cite{FunkeStorandt2016} cite space efficiency, the speed of arithmetic operations, and stability as advantages of low-precision edge weights.

However, as argued and demonstrated by Funke and Storandt~\cite{FunkeStorandt2016,Storandt2018}, naively rounding weights to the nearest integer values could lead to rounding errors accumulating in such a way, that the structure of optimal paths in the graph changes, which can be highly undesirable. When rounding weights naively, some paths may see their lengths doubled whereas other, arbitrarily long paths may see their lengths reduced to zero~\cite{FunkeStorandt2016}. Funke and Storandt argue that randomized rounding is also likely to cause unacceptable errors in any graph that is large enough~\cite{FunkeStorandt2016}.

This brings us to the following problem statement. Consider an undirected graph, denoted by $\graph = (\vertexSet, \edgeSet, \weights)$, with vertex set $\vertexSet$, edge set $\edgeSet$, and a weight function $\weights: \edgeSet \to \mathbb{R}_{\geq 0}$. A \emph{simple} path in $\graph$ is a sequence $\pi$ of distinct vertices $v_1, \ldots, v_j$, where $\lbrace v_i, v_{i+1} \rbrace \in\edgeSet$ for $1 \leq i \leq j-1$. By $\weights(\pi)$ we denote the weight of the path $\pi$, that is, $\sum_{i=1}^{j-1} \weights(\lbrace v_i, v_{i+1} \rbrace)$. A \emph{shortest} path in $\graph$ is a simple path $v_1, \ldots, v_j$ that has minimum  weight among all paths from $v_1$ to $v_j$ in $\graph$.

\begin{definition}[path-oblivious/weak/strong \rounding]
Let $\graph = (\vertexSet, \edgeSet, \weights)$ be an undirected graph with a weight function $\weights: \edgeSet \to \mathbb{R}_{\geq 0}$. We call $\roundedWeights : \edgeSet \to \mathbb{N}_0$ a \emph{path-oblivious \rounding} on \graph if the following condition holds:
\begin{enumerate}
\item \label{condition:weights}
For any shortest path $\pi$ in $\graph$, we have $\abs{\roundedWeights(\pi) - \weights(\pi)} < \eps$, that is, between $\weights$ and $\roundedWeights$, the weight of any shortest path in \graph changes by strictly less than~$\eps$.
\end{enumerate}
We call $\roundedWeights$ a \emph{weak \rounding} if in addition, the following condition holds:
\begin{enumerate}\setcounter{enumi}{1}
\item \label{condition:weakmaintenance}
Any shortest path in $\graph = (\vertexSet, \edgeSet, \weights)$ is also a shortest path in $\roundedGraph = (\vertexSet, \edgeSet, \roundedWeights)$.
\end{enumerate} 
Moreover, we call \roundedWeights a \emph{strong \rounding} if it is a weak \rounding and additionally the following condition is satisfied:
\begin{enumerate}\setcounter{enumi}{2}
\item \label{condition:strongmaintenance}
Any shortest path in $\roundedGraph$ is also a shortest path in $\graph$.
\end{enumerate}
\end{definition}

Of course, for any finite undirected weighted graph there is an $\eps$ such that the graph admits a trivial weak $\eps$-rounding: we could simply choose $\eps$ to be larger than the diameter of the graph and round all weights down to zero\footnote{We do not know whether there is always an $\eps$ such that the graph admits a \emph{strong} $\eps$-rounding.}.\HH{See footnote. We should this give some more thought: it must be possible to prove or disprove this for strong $\eps$-roundings.} However, that would make the concept of rounding moot. We would rather have an \rounding for a small value of $\eps$ such as $\eps = 1$, but in that case, an \rounding does not always exist. For example, a star that consists of three edges of weight $1/2$ does not admit a \rounding[1]: at least two of the three edges would have to be rounded in the same way, but if we would round two edges down, there would be a shortest path with rounding error $-1$; if we would round two edges up, there would be a shortest path with rounding error 1.
Given an undirected graph $\graph = (\vertexSet, \edgeSet, \weights)$ with non-negative real weight function $\weights: \edgeSet \to \mathbb{R}_{\geq 0}$ and a error tolerance \eps, our problem is therefore to decide whether \graph admits a path-oblivious, weak, or strong \rounding. 

Our first goal in solving these problems is to establish for what classes of graphs efficient exact algorithms may exist. In this paper, we show that all three versions of the decision problem are NP-hard for general graphs, but can be solved in quadratic time on trees. In fact, trees always admit a 2-rounding, and given a tree of $n$ vertices, we can compute, in $O(n^2 \log n)$ time, the smallest $\eps$ such that the tree admits an $\eps'$-rounding for any $\eps' > \eps$. The algorithm is constructive, that is, it can easily be adapted to produce the corresponding weights $\roundedWeights$. We compare our results to related work and discuss directions for further research in the last section of the paper.

%% file: includes/complexity.tex
\newcommand{\NP}{\textsc{NP}}
\newcommand{\threeSAT}{\textsc{3-SAT}\xspace}
\newcommand{\threeCNF}{\textsc{3-CNF}\xspace}
\newcommand{\reductGraph}{\ensuremath{G_\alpha}\xspace}
\newcommand{\formula}{\ensuremath{\alpha}\xspace}
\newcommand{\mapping}{\ensuremath{\psi}\xspace}
\newcommand{\false}{\texttt{false}\xspace}
\newcommand{\true}{\texttt{true}\xspace}
\newcommand{\variableVertex}{\ensuremath{v}\xspace}
\newcommand{\nvariableVertex}{\ensuremath{\overline{v}}\xspace}
\newcommand{\clauseVertex}{\ensuremath{c}\xspace}
\newcommand{\diameter}{\ensuremath{D}\xspace}
\newcommand{\myTriangle}[1][]{\ensuremath{\Delta_{i,#1}}\xspace}
\newcommand{\cvEdge}{clause-variable edge\xspace}
\newcommand{\scEdge}{shortcut edge\xspace}
\newcommand{\cvEdges}{clause-variable edges\xspace}
\newcommand{\scEdges}{shortcut edges\xspace}
\newcommand{\vGadget}{variable gadget\xspace}
\newcommand{\vGadgets}{variable gadgets\xspace}
\newcommand{\cGadget}{clause gadget\xspace}
\newcommand{\cGadgets}{clause gadgets\xspace}
\newcommand{\literal}{\ensuremath{l}\xspace}
\newcommand{\edgeOf}[1]{\ensuremath{e(#1)}\xspace}

\section{Complexity of the problem}
We prove the \NP-hardness of the problem by reduction from \threeSAT.
More precisely, we will show that it is \NP-hard to decide, given an edge-weighted graph \graph and an error tolerance $\eps$, whether \graph admits a path-oblivious, weak, or strong \rounding.
For simplicity, we use $\eps = 1$, but the proof is easily adapted to any $\eps \in (7/8, 1]$.
The reduction we present, proves hardness for all three variants of the problem.

The \threeSAT problem is the following. We are given a \threeCNF formula, that is, a boolean formula \formula in conjunctive normal form, where each of the $m$ clauses consists of exactly three literals. Each literal is either one of $n$ variables $x_1, x_2, \ldots, x_n$ or its negation.
Decide whether \formula is satisfiable.
\Wlog we assume that every variable appears at most once in each clause of the \threeSAT formula.%
\footnote{%
Otherwise, we transform \formula in polynomial time as follows: we first remove any clauses that contain a variable and its negation, since they are always satisfied. Next we introduce three new variables $a$, $b$, and $c$. In any clause that contains two or three copies of the same literal, we replace the second copy by $b$ and the third copy (if it exists) by $c$. Finally we add clauses 
$(\neg a \vee \neg b \vee \neg c),
 (\neg a \vee \neg b \vee      c),
 (\neg a \vee      b \vee \neg c),
 (     a \vee \neg b \vee \neg c),
 (     a \vee \neg b \vee      c),
 (     a \vee      b \vee \neg c)$.
This ensures that every truth assignment to $a, b$ and $c$ makes at least one clause false, unless $b$ and $c$ are both false, and therefore do not affect the original clauses in which they replace a duplicate literal.}

In the following, we show how to construct a graph $\reductGraph = (V, E, \weights)$ for a given \threeCNF formula \formula such that \reductGraph admits a strong \rounding[1] if \formula is satisfiable, whereas \reductGraph does not even admit a path-oblivious \rounding[1] if \formula is not satisfiable.
To describe \reductGraph, we introduce subgraphs called \emph{\vGadgets} and \emph{\cGadgets}, as well as \emph{\cvEdges} and \emph{\scEdges}.

The idea of the construction is as follows.
In \autoref{lemma:variableGadget}, we will show that a \vGadget admits exactly two strong \roundings[1].
We identify these two roundings with the assignments \true and \false of a boolean variable.
Using \cvEdges, the state of a \vGadget can be transferred to a \cGadget (\autoref{lemma:clauseHandles}).
Locally, the \cGadget admits a \rounding[1] if and only if one of the variable assignments (transferred via \cvEdges) satisfies the clause (\autoref{lemma:clauseGadget}).
We use \scEdges to ensure that shortest paths in \reductGraph that do not contribute to modelling \formula are easy to analyse and unique---before and after rounding the weights (\autoref{lemma:shortcutEdge}).

To design a \emph{\vGadget}, first consider two edges attached to a triangle, where each edge is of weight $2.5$, as illustrated in \autoref{figure:variableGadget:initial}.
In a \rounding[1], the choice of the rounding for \edgeOf{\variableVertex_{i,0}}, the edge incident on $\variableVertex_{i,0}$, determines the rounding of the remaining edges; see  \autoref{figure:variableGadget:roundedDown} and \autoref{figure:variableGadget:roundedUp}.
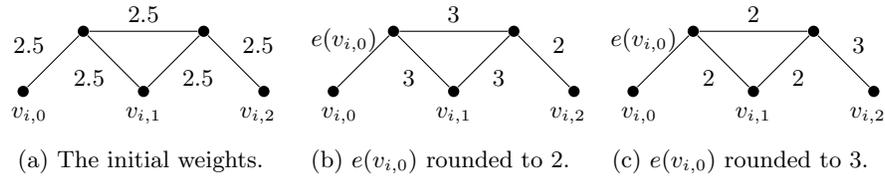
\begin{figure}[tb]
\centering
\begin{subfigure}[b]{0.3\textwidth}
\centering
\begin{tikzpicture}[scale=.8, vertex/.style={circle, fill, inner sep=1.5pt}]
\coordinate (v0) at (0, 0);
\node[vertex] (col{0}v2) at (v0) {};

\foreach \i in {2}{
	\node[vertex] (col{\i}v1) at ($(v0) + (\i - 1, 1)$) {};
	\node[vertex] (col{\i}v2) at ($(v0) + (\i + 1, 1)$) {};
	\node[vertex] (col{\i}v3) at ($(v0) + (\i 	, 	0)$) {};

	\draw (col{\i}v1) to node[midway, above left] {2.5} (col{\number\numexpr \i - 2 \relax}v2);	
	\draw (col{\i}v1) to node[midway, above] {2.5} (col{\i}v2);
	\draw (col{\i}v1) to node[midway, below left] {2.5} (col{\i}v3);
	\draw (col{\i}v2) to node[midway, below right] {2.5} (col{\i}v3);
}

\node[vertex] (col{4}v2) at (4, 0) {};
\draw (col{4}v2) to node[midway, above right] {2.5} (col{2}v2);

\node at ($(col{0}v2) - (-0.1, .4)$) {$\variableVertex_{i,0}$};
\node at ($(col{2}v3) - (0, .4)$) {$\variableVertex_{i,1}$};
\node at ($(col{4}v2) - (0.1, .4)$) {$\variableVertex_{i,2}$};
\end{tikzpicture}
\caption{The initial weights.}
\label{figure:variableGadget:initial}
\end{subfigure}
~
\begin{subfigure}[b]{0.3\textwidth}
\centering
\begin{tikzpicture}[scale=.8, vertex/.style={circle, fill, inner sep=1.5pt}]
\coordinate (v0) at (0, 0);
\node[vertex] (col{0}v2) at (v0) {};

\foreach \i in {2}{
	\node[vertex] (col{\i}v1) at ($(v0) + (\i - 1, 1)$) {};
	\node[vertex] (col{\i}v2) at ($(v0) + (\i + 1, 1)$) {};
	\node[vertex] (col{\i}v3) at ($(v0) + (\i 	, 0)$) {};

	\draw (col{\i}v1) to (col{\number\numexpr \i - 2 \relax}v2);	
	\draw (col{\i}v1) to node[midway, above] {3} (col{\i}v2);
	\draw (col{\i}v1) to node[midway, below left] {3} (col{\i}v3);
	\draw (col{\i}v2) to node[midway, below right] {3} (col{\i}v3);
}

\node[vertex] (col{4}v2) at (4, 0) {};
\draw (col{4}v2) to node[midway, above right] {2} (col{2}v2);

\node at ($(col{0}v2) + (0.18, .85)$) {\edgeOf{\variableVertex_{i,0}}};
\node at ($(col{0}v2) - (-0.1, .4)$) {$\variableVertex_{i,0}$};
\node at ($(col{2}v3) - (0, .4)$) {$\variableVertex_{i,1}$};
\node at ($(col{4}v2) - (0.1, .4)$) {$\variableVertex_{i,2}$};
\end{tikzpicture}
\caption{\edgeOf{\variableVertex_{i,0}} rounded to 2.}
\label{figure:variableGadget:roundedDown}
\end{subfigure}
~
\begin{subfigure}[b]{0.3\textwidth}
\centering
\begin{tikzpicture}[scale=.8, vertex/.style={circle, fill, inner sep=1.5pt}]
\coordinate (v0) at (0, 0);
\node[vertex] (col{0}v2) at (v0) {};

\foreach \i in {2}{
	\node[vertex] (col{\i}v1) at ($(v0) + (\i - 1, 1)$) {};
	\node[vertex] (col{\i}v2) at ($(v0) + (\i + 1, 1)$) {};
	\node[vertex] (col{\i}v3) at ($(v0) + (\i 	, 0)$) {};

	\draw (col{\i}v1) to (col{\number\numexpr \i - 2 \relax}v2);	
	\draw (col{\i}v1) to node[midway, above] {2} (col{\i}v2);
	\draw (col{\i}v1) to node[midway, below left] {2} (col{\i}v3);
	\draw (col{\i}v2) to node[midway, below right] {2} (col{\i}v3);
}

\node[vertex] (col{4}v2) at (4, 0) {};
\draw (col{4}v2) to node[midway, above right] {3} (col{2}v2);

\node at ($(col{0}v2) + (0.18, .85)$) {\edgeOf{\variableVertex_{i,0}}};
\node at ($(col{0}v2) - (-0.1, .4)$) {$\variableVertex_{i,0}$};
\node at ($(col{2}v3) - (0, .4)$) {$\variableVertex_{i,1}$};
\node at ($(col{4}v2) - (0.1, .4)$) {$\variableVertex_{i,2}$};
\end{tikzpicture}
\caption{\edgeOf{\variableVertex_{i,0}} rounded to 3.}
\label{figure:variableGadget:roundedUp}
\end{subfigure}
\caption{
A minimal \vGadget.
In a \rounding[1] of this gadget, any shortest path of two edges has to round one edge up and the other edge down---otherwise the total rounding error on the path would be $\pm1$, violating Condition~\ref{condition:weights} of a \rounding[1].
Thus, the top triangle edge must be rounded in the opposite way as compared to the edges \edgeOf{\variableVertex_{i,0}} and \edgeOf{\variableVertex_{i,2}}, incident on $\variableVertex_{i,0}$ and $\variableVertex_{i,2}$, respectively.
These arguments imply that \edgeOf{\variableVertex_{i,0}} and \edgeOf{\variableVertex_{i,2}} have to be rounded in the same way and \emph{all} triangle edges are rounded in the opposite way.
Note that paths containing two triangle edges have rounding error $\pm 1$, but such paths are not shortest paths, neither before nor after rounding.
}
\label{figure:minimalVariableGadget}
\end{figure}
To obtain a \vGadget for variable $x_i$, we proceed as follows.
Assume that $x_i$ appears in $h$ literals $\literal_1, \ldots, \literal_{h}$ of \formula.
We construct $h$ triangles $\myTriangle[1], \ldots, \myTriangle[h]$, where each $\myTriangle[k]$ (for $k \in \{1,...,h\}$) has a left vertex, a right vertex, and a \emph{base} (bottom) vertex; we label the base vertex $\variableVertex_{i,k}$. We chain up the triangles by including an edge between the right vertex of $\myTriangle[k]$ and the left vertex of $\myTriangle[k+1]$ for each $k \in \{1,...,h-1\}$. To the left vertex of $\myTriangle[1]$, we attach another vertex $\variableVertex_{i,0}$, and to the right vertex of $\myTriangle[h]$, we attach another vertex $\variableVertex_{i,h+1}$, as shown in \autoref{figure:variableGadget}.
Finally, for every $1\leq k \leq h$, if $\literal_k = \neg x_i$, we add another vertex $\overline{\variableVertex}_{i,k}$, called \emph{inverter}, which we connect to $\variableVertex_{i,k}$.
All edges of the \vGadget have an initial weight of 2.5.

\begin{figure}[tb]
\centering
\begin{tikzpicture}[scale=0.7, vertex/.style={circle, fill, inner sep=1.5pt}]
\coordinate (v0) at (0, 0);
\node[vertex] (col{0}v2) at (v0) {};

\foreach \i in {2, 6, 12}{
	\node[vertex] (col{\i}v1) at ($(v0) + (\i-1, 1)$) {};
	\node[vertex] (col{\i}v2) at ($(v0) + (\i+1, 1)$) {};
	\node[vertex] (col{\i}v3) at ($(v0) + (\i, 	0)$) {};

	\draw (col{\i}v1) to (col{\i}v2);
	\draw (col{\i}v1) to (col{\i}v3);
	\draw (col{\i}v2) to (col{\i}v3);
}

\node[vertex] (col{14}v2) at (14, 0) {};
\node[vertex] (col{6}v4) at (6, -1) {};
\node at ($(col{0}v2) - (.3,  .3)$) {$\variableVertex_{i,0}$};
\node at ($(col{2}v3) + (0,  .6)$) {\myTriangle[1]};
\node at ($(col{2}v3) - (.3,  .3)$) {$\variableVertex_{i,1}$};
\node at ($(col{6}v3) + (0,  .6)$) {\myTriangle[2]};
\node at ($(col{6}v3) - (.4,  .3)$) {$\variableVertex_{i,2}$};
\node at ($(col{6}v4) - (.4,  .3)$) {$\overline{\variableVertex}_{i,2}$};
\node at ($(col{12}v3) + (0,  .6)$) {\myTriangle[h]};
\node at ($(col{12}v3) - (.3,  .3)$) {$\variableVertex_{i,h}$};
\node at ($(col{14}v2) - (.3,  .3)$) {$\variableVertex_{i,h+1}$};

\node at (9,.5) {$\ldots$};
\draw (col{0}v2) to node[midway, above left] {\edgeOf{\variableVertex_{i,0}}} (col{2}v1);
\draw (col{2}v2) to (col{6}v1);
\draw (col{6}v2) to ($(v0) + (7.5,1)$);
\draw (col{12}v1) to ($(v0) + (10.5,1)$);
\draw (col{14}v2) to (col{12}v2);
\draw (col{6}v3) to (col{6}v4);
\end{tikzpicture}
\caption{
The \vGadget for $x_i$, where $x_i$ appears in $h$ literals $\literal_1, \ldots, \literal_{h}$ of~\formula.
Here, $\literal_2 = \neg x_i$, so an additional vertex $\overline{\variableVertex}_{i,2}$ is added and attached to $\variableVertex_{i,2}$.
All edges have weight 2.5, so the choice of the rounding for \edgeOf{\variableVertex_{i,0}} determines the rounding of all the other edges in a \rounding[1]:
triangle edges have to be rounded complementary to non-triangle edges.
}
\label{figure:variableGadget}
\end{figure}
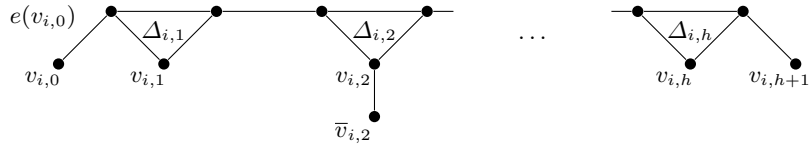

We call the edges of $\myTriangle[1], \ldots, \myTriangle[h]$ \emph{triangle edges}.
Moreover, with \edgeOf{\variableVertex_{i,0}} we denote the unique edge of the \vGadget attached to $\variableVertex_{i,0}$.
A \vGadget has only two different \roundings[1]:
\begin{lemma}[{\roundings[1]} of \vGadgets]\label{lemma:variableGadget}
A \vGadget admits exactly two \roundings[1] (both of which are strong \roundings[1]):
either all triangle edges are rounded up and all other edges down, or vice versa.

Moreover, for any two vertices $u, v$ of the gadget, the rounding error of the unique shortest path from $u$ to $v$ is either zero or equal to the rounding error on the last edge of the path ending at $v$.
\end{lemma}
\begin{proof}
For a single triangle and its adjacent non-triangle edges, \autoref{figure:minimalVariableGadget} explains that all triangle edges must be rounded in the same way, while non-triangle edges must be rounded in the opposite way.
Tracking the triangles in \autoref{figure:variableGadget} from left to right, the same holds for any number of triangles, by induction.

We will now argue that the two roundings that can be obtained in this way, indeed satisfy the conditions of a strong \rounding[1].
Note that, before and after rounding, paths with two consecutive triangle edges are no shortest paths.
Thus, a simple path in this gadget is a shortest path before rounding if and only if it is a shortest path after rounding, and in any shortest path triangle and non-triangle edges alternate.
Therefore, any shortest path with an even number of edges has rounding error zero; any shortest path with an odd number of edges has the rounding error of its last edge.
This establishes Conditions~\ref{condition:weights}, 
\ref{condition:weakmaintenance} and \ref{condition:strongmaintenance} of a strong \rounding[1] and thus completes the proof of the lemma. 
\end{proof}

From \autoref{lemma:variableGadget}, we obtain that in a \rounding[1], the choice of the rounding for \edgeOf{\variableVertex_{i,0}} determines the rounding of all other edges.

To create a \textit{\cGadget} for clause $C_j$, we take a cycle of nine vertices and nine edges, where each edge gets an initial weight of $3.6$.
Moreover, we attach, to every third vertex along the cycle, a new vertex, called a \emph{knob}, with another edge of weight $2.5$, called a \emph{handle}.
We denote the knobs by $\clauseVertex_{j,1}, \clauseVertex_{j, 2}, \clauseVertex_{j, 3}$, as shown in \autoref{figure:clauseGadget:initial}.
We will use the notation \edgeOf{\clauseVertex_{j,t}} to denote the edge (handle) of the \cGadget that connects $\clauseVertex_{j,t}$ to the nonagon.
Finally, we add a vertex which we connect to every vertex on the cycle with an edge of weight $6$.
Note that the weights of these edges are integer---hence, they cannot be rounded.
A \cGadget has at least three strong \roundings[1]; see \autoref{figure:clauseGadget:rounded}.
However, there is no path-oblivious, weak, or strong \rounding[1] for the clause gadget in which \edgeOf{\clauseVertex_{j,1}}, \edgeOf{\clauseVertex_{j,2}} and \edgeOf{\clauseVertex_{j,3}} are all rounded up, as the following lemma states.

\begin{figure}
\centering
\begin{subfigure}[b]{0.39\textwidth}
	\centering
	\begin{tikzpicture}[-, auto, vertex/.style={circle, fill, inner sep=1.5pt}]
	\clauseGadget{}
	\clauseGadgetDrawLabels{}
	\clauseRounding[.9]{3.6,3.6,3.6,3.6,3.6,3.6,3.6,3.6,3.6}
	\node at ($(c) + (140:1.5)$) {2.5};
	\node at ($(c) + ( 40:1.5)$) {2.5};
	\node at ($(c) + (-80:1.5)$) {2.5};
	\end{tikzpicture}
    \caption{The initial weight of all edges incident on the centre is 6.}
    \label{figure:clauseGadget:initial}
\end{subfigure}
~
\begin{subfigure}[b]{0.58\textwidth}
	\centering
	\begin{tikzpicture}[-, auto, vertex/.style={circle, fill, inner sep=1.5pt}]
	\clauseGadget{}
	\clauseRounding[0.9]{4,3,4, 4,3,4, 4,3,4}
	\node at ($(c) + (140:1.5)$) {2};
	\node at ($(c) + ( 40:1.5)$) {2};
	\node at ($(c) + (-70:1.6)$) {2 or 3};
	\end{tikzpicture}
	~
	\begin{tikzpicture}[-, auto, vertex/.style={circle, fill, inner sep=1.5pt}]
	\clauseGadget{}
	\clauseRounding[0.9]{4,3,4, 3,4,3, 4,3,4}
	\node at ($(c) + (140:1.5)$) {3};
	\node at ($(c) + ( 40:1.5)$) {3};
	\node at ($(c) + (-80:1.5)$) {2};
	\end{tikzpicture}
    \caption{
    Three \roundings[1] of a \cGadget. No two consecutive edges along the nonagon can be rounded down in a \rounding[1]. At most two of the edges $\edgeOf{\clauseVertex_{j,1}}, \edgeOf{\clauseVertex_{j,2}}$ and $\edgeOf{\clauseVertex_{j,3}}$ are rounded up.
    }
    \label{figure:clauseGadget:rounded}
\end{subfigure}
\caption{The \cGadget for clause $C_j$.
}
\label{figure:clauseGadget}
\end{figure}
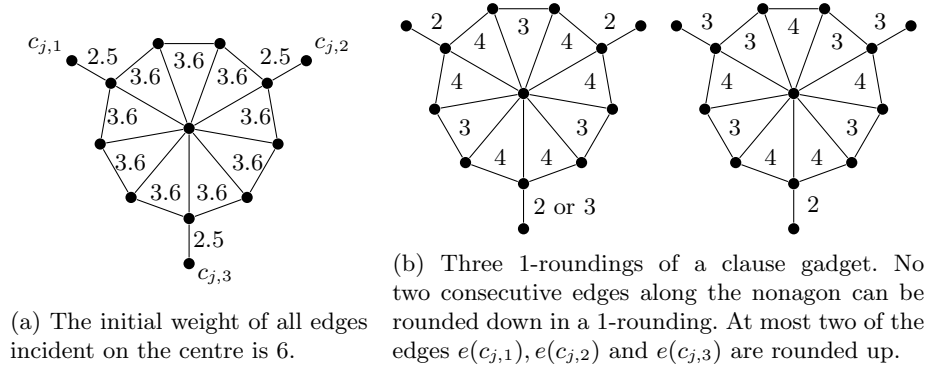
\begin{lemma}[{\roundings[1]} of a \cGadget]\label{lemma:clauseGadget}
Consider a \cGadget for~$C_j$, and suppose we fix, for each of its handles $\edgeOf{\clauseVertex_{j,1}}, \edgeOf{\clauseVertex_{j,2}}$ and $\edgeOf{\clauseVertex_{j,3}}$, whether its weight is rounded up or down.
The \cGadget now admits a path-oblivious \rounding[1] \roundedWeights if and only if at least one of its three handles is rounded down.
If there is a path-oblivious \rounding[1], there is a strong \rounding[1].
\end{lemma}
\begin{proof}
\autoref{figure:clauseGadget:rounded} shows a strong \rounding[1] for the \cGadget for the cases in which one, two, or three of the edges $\edgeOf{\clauseVertex_{j,1}}, \edgeOf{\clauseVertex_{j,2}}$ and $\edgeOf{\clauseVertex_{j,3}}$ are rounded down.
It remains to show that if none of these three edges are rounded down, it is impossible to obtain a \rounding[1] (not even a path-oblivious one) for the complete gadget.

Indeed, assume to the contrary that the weights of \edgeOf{\clauseVertex_{j,1}}, \edgeOf{\clauseVertex_{j,2}}, and \edgeOf{\clauseVertex_{j,3}} are all rounded up.
For the shortest path between any pair of $\clauseVertex_{j,1}, \clauseVertex_{j,2}, \clauseVertex_{j,3}$, we now know that (at least) two of the three edges along the nonagon have to be rounded down. 
Otherwise the shortest path of length $2\cdot 2.5 + 3 \cdot 3.6 = 15.8$ with respect to \weights would have weight (at least) $2\cdot 3 + 1 \cdot 3 + 2 \cdot 4 = 17$ with respect to \roundedWeights, contradicting Condition~\ref{condition:weights} of a \rounding[1].
However, if the weight of at least six of the nine edges along the nonagon has to be rounded down, then there have to be two adjacent edges $\lbrace u, v \rbrace, \lbrace v, w \rbrace$, each of weight 3.6, that are rounded down.
Consequently, the shortest path from $u$ to $w$ has weight 7.2 with respect to \weights and 6 with respect to \roundedWeights, which again contradicts Condition~\ref{condition:weights} of a \rounding[1].
\end{proof}
Next, we introduce \emph{\cvEdges} to connect the gadgets of variables to the gadgets of clauses that contain these variables. 
So if variable $x_i$ appears in clause $C_j$, we connect the corresponding gadgets with exactly one \cvEdge of weight $\diameter := 5m + 20$ according to the following rule:
if the $t$-th literal in $C_j$ is $x_i$, then we connect a base vertex of the \vGadget for $x_i$ to $\clauseVertex_{j,t}$, using an edge of weight \diameter;
if the $t$-th literal in $C_j$ is $\neg x_i$, then we connect an inverter of the \vGadget for $x_i$ to $\clauseVertex_{j,t}$, using an edge of weight \diameter.
We do this such that exactly one \cvEdge is connected to each inverter, and exactly one \cvEdge is connected to each base vertex that is not attached to an inverter.
By design, the \vGadgets have the right numbers of base vertices and inverters to make this possible.

Note that \cvEdges do not invalidate Lemmas~\ref{lemma:variableGadget} and \ref{lemma:clauseGadget}, that is, they still hold with respect to the shortest paths between any pair of vertices of the variable or clause gadget, respectively. This can be seen as follows.
There are $m$ clauses and each variable appears in each clause at most once.
Hence, the diameter of a \vGadget is at most $(2m + 1) \cdot 2.5$. The diameter of a \cGadget is 15.8.
Thus, the variable and \cGadgets all have diameter less than $D-2$.
Therefore, before rounding, no path between two vertices of the same gadget that uses a clause-variable edge can be a shortest path.
Moreover, when we choose a \rounding[1] for each gadget separately, the rounded weight of the shortest path between any pair of vertices within a gadget will still be less than $D-1$, while the rounded weight of any path using a \cvEdge will still be at least $D$.
Therefore, also after rounding, no path between two vertices of the same gadget that uses a clause-variable edge can be a shortest path.
Thus, adding the \cvEdges does not invalidate Lemmas~\ref{lemma:variableGadget} and \ref{lemma:clauseGadget}.

Further note that due to this construction, the choice for \edgeOf{\variableVertex_{i,0}} also determines the rounding for \edgeOf{\clauseVertex_{j,t}} in a \rounding[1]:
\begin{lemma}[\cvEdges and {\roundings[1]}]\label{lemma:clauseHandles}
For any \rounding[1] on \reductGraph:\begin{itemize}
\item if $\clauseVertex_{j,t}$ is connected to a base vertex of the \vGadget for $x_i$, then $\edgeOf{\clauseVertex_{j,t}}$ is rounded in the same way as $\edgeOf{\variableVertex_{i,0}}$;
\item if $\clauseVertex_{j,t}$ is connected to an inverter vertex of the \vGadget for $x_i$, then $\edgeOf{\clauseVertex_{j,t}}$ is rounded in the opposite way as $\edgeOf{\variableVertex_{i,0}}$.
\end{itemize}
\end{lemma}
\begin{proof}
Consider the $t$-th literal \literal of clause $C_j$ ($t \in \{1,2,3\}$).
If $\literal = x_i$, then a \cvEdge connects $\clauseVertex_{j,t}$ to a base vertex of the \vGadget for~$x_i$.
Now consider the path that consists of a triangle edge incident to this base vertex (with weight 2.5), the \cvEdge (with integer weight $D$), and $\edgeOf{\clauseVertex_{j,t}}$ (with weight 2.5).
This path has length $D+5$ and is a shortest path, since any other simple path between the same end points would have to make a detour in the \vGadget and have length at least $D+7.5$, or lead over at least three other \cvEdges and have length at least $3D$.
Because the path has integer length, in any \rounding[1] it must have the same length.
It follows that \edgeOf{\clauseVertex_{j,t}} must be rounded in the opposite way as compared to the triangle edge in the \vGadget, which, by \autoref{lemma:variableGadget}, implies that \edgeOf{\clauseVertex_{j,t}} is rounded like the non-triangle edge \edgeOf{\variableVertex_{i,0}}.

Otherwise, $\literal = \neg x_i$ and a \cvEdge connects $\clauseVertex_{j,t}$ to an inverter. Now we consider the path that consists of the non-triangle edge incident on the inverter in the \vGadget, the \cvEdge, and $\edgeOf{\clauseVertex_{j,t}}$. Again, the last edge must be rounded in the opposite way as compared to the first, which by \autoref{lemma:variableGadget}, is rounded in the same way as \edgeOf{\variableVertex_{i,0}}.
\end{proof}

Finally, to ensure that the shortest path between any pair of vertices of \reductGraph is unique (and easy to analyse), we add \emph{\scEdges} according to the following rule. If $u$ and $v$ are vertices of \reductGraph such that one of the following conditions holds:
\begin{enumerate}[label=(\roman*)]
\item $u$ and $v$ belong to different \vGadgets;
\item $u$ and $v$ belong to different \cGadgets;
\item $u$ belongs to a \vGadget for variable $x_i$ and $v$ belongs to a \cGadget for clause $C_j$ and neither $x_i$ nor $\neg x_i$ appears in $C_j$;
\end{enumerate}
then we include an edge $\lbrace u, v \rbrace$ in \reductGraph with weight $2\diameter$.
\begin{lemma}[shortest path via \scEdge]\label{lemma:shortcutEdge}
Let $u$ and $v$ be vertices of \reductGraph that are directly connected by a \scEdge, and let \roundedWeights be a \rounding[1] on \reductGraph.
Then, the \scEdge $\lbrace u, v \rbrace$ is the unique shortest path in \reductGraph with respect to \weights and \roundedWeights.
\end{lemma}
\begin{proof}
By construction, $u$ and $v$ belong to different gadgets and there is no \cvEdge between these gadgets.
Therefore, any path from $u$ to $v$ other than the direct \scEdge $\lbrace u, v \rbrace$, must use (a) another \scEdge of weight $2\diameter$ plus at least one other edge of weight at least 2.5, or (b) at least two \cvEdges of weight $\diameter$ each, plus at least one other edge of weight at least 2.5 (because no vertex is incident on more than one \cvEdge).
In both cases, before and after rounding, the total weight of the path would be at least 2 more than the weight of $\{u,v\}$, which is $2\diameter$. Therefore, $\{u,v\}$ is the unique shortest path in \reductGraph with respect to both \weights and \roundedWeights.
\end{proof}
Note that, just like \cvEdges, the \scEdges do not invalidate Lemmas \ref{lemma:variableGadget} and~\ref{lemma:clauseGadget}. They do not invalidate \autoref{lemma:clauseHandles} either, as its proof hinges on shortest paths of length $\diameter+5 < 2\diameter-2$.
An example for the construction is given in \autoref{figure:reductGraph}.
\begin{figure}
\centering
\includegraphics[width=.9\textwidth]{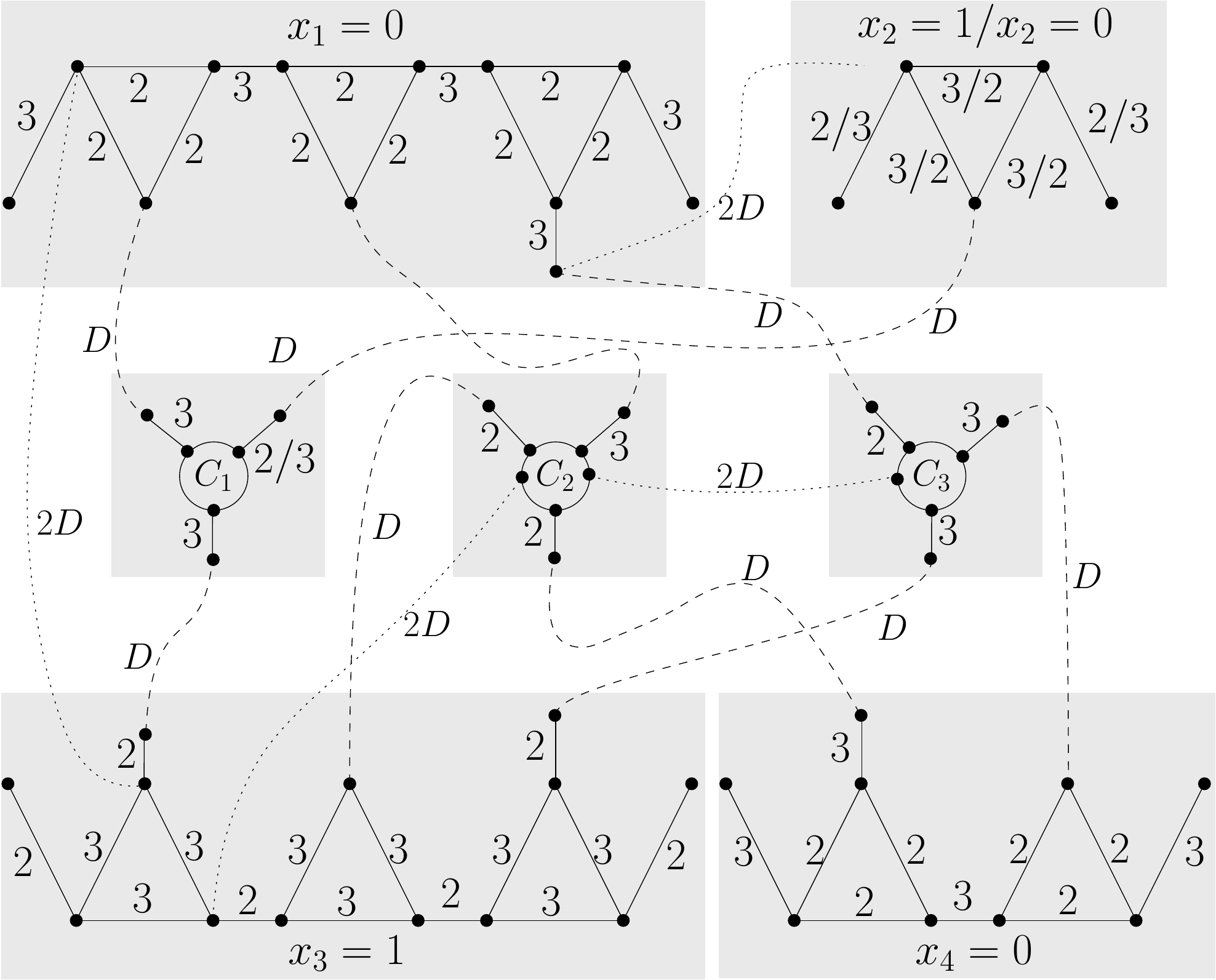}
\caption{
A sketch of \reductGraph for $\formula = (x_1 \vee x_2 \vee \neg x_3) \wedge (x_1 \vee x_3 \vee \neg x_4) \wedge (\neg x_1 \vee \neg x_3 \vee x_4)$.
Grey bounding boxes mark gadgets, \cvEdges are dashed and have weight $\diameter = 5\cdot 3 + 20 = 35$, \scEdges (mostly omitted) are dotted and have weight $2\diameter$.
The weights of the edges in \reductGraph have been rounded according to the assignment of the variables, which are given in the boxes of the corresponding \vGadgets.
For $x_2$, both assignments are given, which affects the rounded weight of several edges.
The rounded weight of each of these edges is given by $a / b$, where $a$ corresponds to the assignment $x_2 = 1$ and $b$ corresponds to $x_2 = 0$.
Note that for $x_1=0,x_3=1,x_4=0$, if we set $x_2 =1$, we can obtain a \rounding[1] for the gadget of $C_1$ and for \reductGraph.
If we set $x_2 = 0$, then there is no \rounding[1] for the gadget of $C_1$, as all of its handles are rounded up.}
\label{figure:reductGraph}
\end{figure}

\begin{theorem}
it is \NP-hard to decide, given an edge-weighted graph \graph and an error tolerance $\eps$, whether \graph admits (1) a path-oblivious \rounding; (2) a weak \rounding; (3) a strong \rounding.
\end{theorem}
\begin{proof}
\HH{We should check if we should insert some parts of the concise proof sketch in the IWOCA version to give more/faster intuition and/or clarify the structure of the proof.}
We prove that if \formula is satisfiable, \reductGraph admits a strong \rounding[1].
Moreover, if \reductGraph admits a path-oblivious \rounding[1], then \formula is satisfiable.

To start with, we show how to obtain a strong \rounding[1] for \reductGraph if \formula is satisfiable.
Let \mapping be an assignment of values to the variables that satisfies \formula.
So we have $\mapping: \lbrace x_1, \ldots, x_n \rbrace \to \lbrace 0, 1 \rbrace$, where $0$ denotes the logical value \false and $1$ denotes \true; correspondingly, $\mapping(\neg x_i) = 1 - \mapping(x_i)$. Then, a strong \rounding[1] \roundedWeights for \reductGraph can be constructed in the following way:
\begin{itemize}
\item 
edges with integer weight keep their weight;
\item
in a \vGadget for $x_i$:
if $\mapping(x_i) = 0$, all triangle edges are rounded down and the non-triangle edges (including \edgeOf{\variableVertex_{i,0}}) are rounded up;
if $\mapping(x_i) = 1$, all triangle edges are rounded up and the non-triangle edges (including \edgeOf{\variableVertex_{i,0}}) are rounded down;
\item 
in a \cGadget for clause $C_j$ with literals $l_1,l_2,l_3$: for $t \in \{1,2,3\}$, if $\mapping(l_t) = 0$, the weight of $\edgeOf{\clauseVertex_{j,t}}$ is rounded up; if $\mapping(l_t) = 1$, the weight of $\edgeOf{\clauseVertex_{j,t}}$ is rounded down. Observe that this implies that if $l_t = x_i$, then $\edgeOf{\clauseVertex_{j,t}}$ is rounded like the non-triangle edges in the \vGadget for $x_i$, and if $l_t = \neg x_i$, then $\edgeOf{\clauseVertex_{j,t}}$ is rounded like the triangle edges.
Since $\formula$ is satisfied, also $C_j$ is satisfied and at least one of the edges \edgeOf{\clauseVertex_{j,1}}, \edgeOf{\clauseVertex_{j,2}}, \edgeOf{\clauseVertex_{j,3}} is rounded down.
By \autoref{lemma:clauseGadget}, we can complete the rounding of the gadget to one of those given in \autoref{figure:clauseGadget:rounded} (modulo rotation).
That is, for each pair of edges from $\edgeOf{\clauseVertex_{j,1}}, \edgeOf{\clauseVertex_{j,2}}$, and $\edgeOf{\clauseVertex_{j,3}}$, if both of them are rounded up, we round the three edges between them along the nonagon to 3,~4,~3; in all other cases, we round the three edges between them along the nonagon to 4,~3,~4.
\end{itemize}
Certainly, the weight of every edge has been rounded either up or down and is now integer.
It remains to prove that Conditions~\ref{condition:weights}, \ref{condition:weakmaintenance} and \ref{condition:strongmaintenance} of a strong \rounding[1] are fulfilled for the shortest paths between any pair of vertices $u$ and $v$ in $\reductGraph$.
Let $u, v$ be vertices of \reductGraph.
If \reductGraph contains a \scEdge $\lbrace u, v \rbrace$, \autoref{lemma:shortcutEdge} applies and the conditions hold.
Otherwise, that is, if $u$ and $v$ are not connected by a \scEdge, we have to distinguish two cases:
\begin{enumerate}[label=(\roman*)]
\item $u$ and $v$ lie in the same gadget;\label{proof:case:sameGadget}
\item $u$ lies in a \vGadget for a variable $x_i$ and $v$ in a \cGadget for a clause $C_j$ and either $x_i$ or $\neg x_i$ appears in $C_j$.\label{proof:case:linkedGadgets}
\end{enumerate}

\ref{proof:case:sameGadget}
For each gadget, we constructed a strong \rounding[1] according to \autoref{lemma:variableGadget} or \ref{lemma:clauseGadget}.
As observed above, the lemmata continue to hold after adding \cvEdges and \scEdges.
Thus, if $u$ and $v$ lie in the same gadget, Conditions~\ref{condition:weights}, \ref{condition:weakmaintenance}, \ref{condition:strongmaintenance} are satisfied with respect to shortest paths between $u$ and $v$.

\ref{proof:case:linkedGadgets}
Since each variable occurs at most once in any clause, the gadgets of $u$ and $v$ are, by construction, connected by exactly one \cvEdge $\{s,c\}$, where $s$ is a base or inverter vertex in the \vGadget, and $c = \clauseVertex_{j, t}$ is a handle of the \cGadget, for some $t$.

\DK{We could use a figure here, too}
We start with verifying Condition~\ref{condition:weights} of the \rounding[1] for a shortest path \shortestPath{u, v} from $u$ to $v$ that uses the \cvEdge $\{s, c\}$.
For any two vertices $a$ and $b$ on \shortestPath{u, v}, let \shortestPath{a, b} be the subpath of \shortestPath{u, v} from $a$ to $b$.
If $C_j$ includes the positive literal~$x_i$, then $s$ is a base vertex of the \vGadget, by construction.
Consider the unique shortest path \shortestPath{u,s} from $u$ to $s$ in the \vGadget.
By \autoref{lemma:variableGadget}, the rounding error on \shortestPath{u,s} is either zero or equal to the rounding error on any triangle edge (which is either $-0.5$ or 0.5).
Adding $\{s,c\}$ to the path does not change the rounding error, as it has integer weight.
If $v = c$, we are done now.
Otherwise, the subpath \shortestPath{c,v} from $c$ to $v$ in the \cGadget uses the edge $\edgeOf{\clauseVertex_{j,t}} = \{c, c'\}$.
By \autoref{lemma:clauseHandles}, \edgeOf{\clauseVertex_{j,t}} is rounded in the same way as the non-triangle edges in the \vGadget.
It follows that the rounding error of either \shortestPath{u, c} or \shortestPath{u, c'} is zero.
The remaining part of the path from $u$ to $v$, that is \shortestPath{c, v} or \shortestPath{c', v}, lies entirely inside the \cGadget and is the unique shortest path to $v$ from $c$ or $c'$, respectively, with rounding error more than $-1$ and less than 1, by \autoref{lemma:clauseGadget}.
Hence, Condition~\ref{condition:weights} is satisfied with respect to $\shortestPath{u, v}$.
If $C_j$ includes the negative literal $\neg x_i$, then $s$ is an inverter vertex, and the whole argument goes through with the roles of triangle and non-triangle edges swapped.

The weight of \shortestPath{u,v} with respect to \weights is bounded by $\weights(\shortestPath{u,v}) = \weights(\shortestPath{u,s}) + \diameter + \weights(\shortestPath{c, v}) \leq 5m + \diameter + 15.8 = 2\diameter - 4.2 \leq 3\diameter - 2$. Since Condition~\ref{condition:weights} is satisfied, the weight with respect to \roundedWeights is therefore less than $3\diameter - 1$.
This bound rules out other paths that do not use $\{s, c\}$:
any other path from $u$ to $v$ that avoids the \cvEdge $\{s,c\}$ must pass by (a) a single other gadget $g$, or (b) at least two other gadgets.
In case (a), the path must use at least two \cvEdges or \scEdges, each of weight at least $\diameter$. In fact, at least one of these edges must be a \scEdge of weight $2D$, because a path from $u$ with two \cvEdges and no \scEdges, could only end in a variable gadget.
In case (b), the path uses at least three \cvEdges or \scEdges, each of weight at least $\diameter$.
In both cases (a) and (b), the total weight of the path would be at least $3\diameter$.
However, such a path cannot be shortest, neither before nor after rounding.

Since the shortest paths within the gadgets are maintained by Lemma~\ref{lemma:variableGadget} and~\ref{lemma:clauseGadget}, it follows that also Conditions \ref{condition:weakmaintenance} and \ref{condition:strongmaintenance} are satisfied with respect to $u$ and $v$, that is, the shortest path between $u$ and $v$ is identical with respect to $\weights$ and $\roundedWeights$.

This completes the proof that if $\formula$ is satisfiable, then \reductGraph admits a strong \rounding[1]. Now we still need to prove that if \reductGraph admits a path-oblivious \rounding[1], then \formula is satisfiable. This we do by constructing, from a given path-oblivious \rounding[1], a choice \mapping of the variables that satisfies \formula.

Assume we are given a \rounding[1] \roundedWeights of \reductGraph.
If in \roundedWeights, the weight of \edgeOf{\variableVertex_{i,0}} is rounded down, we set $\mapping(x_i) = 1$, otherwise we set $\mapping(x_i) = 0$.
Now consider the \cGadget for any clause $C_j$.
Following \autoref{lemma:clauseGadget}, we know that there is at least one $t \in \{1,2,3\}$ such that \edgeOf{\clauseVertex_{j,t}} is rounded down.
By construction, a \cvEdge connects $\clauseVertex_{j,t}$ to a base vertex or an inverter of the \vGadget for a variable $x_i$: a base vertex if $x_i$ appears as a literal in $C_j$, and an inverter if $\neg x_i$ appears as a literal in $C_j$.
Following \autoref{lemma:clauseHandles}, in the first case \edgeOf{\variableVertex_{i,0}} is rounded in the same way as \edgeOf{\clauseVertex_{j,t}}, that is, down, and thus $\mapping(x_i) = 1$ and the literal $x_i$ makes $C_j$ true.
In the second case, \edgeOf{\variableVertex_{i,0}} is rounded in the opposite way of \edgeOf{\clauseVertex_{j,t}}, that is, up, and thus $\mapping(x_i) = 0$, and the literal $\neg x_i$ makes $C_j$ true. The same argument applies to each clause $C_j$, and thus, $\mapping$ satisfies $\formula$.

To complete the proof, observe that \threeSAT is an \NP-hard problem and
\reductGraph consists of $O(mn)$ vertices and $O(m^2n^2)$ edges and can be constructed in polynomial time.
\end{proof}

We observe that the NP-hardness construction as described works for $\eps = 1$ and some smaller values.
The weight $w$ of a nonagon edge in the clause gadget is deciding.
If the \threeCNF formula is satisfiable, the minimum and maximum rounding errors are $10-3w$ and $4\frac12-w$.
With $w= 3.6$, we obtain a rounding error of $\max (\abs{10-3w}, \abs{4 \frac12 - w}) = 0.9$.
The expression is actually minimized to $7/8$ when we choose $w = 3\frac 58$. 
If we choose $\eps> 1$, \autoref{lemma:variableGadget} will not hold.
Thus, the construction works as long as $7/8 < \eps \leq 1$. \HH{If you make the weights on the nonagon edges in the clause gadget 3.75, 3.25, 3.75, 3.75, 3.25, 3.75, 3.75, 3.25, 3.75,
then I believe the construction actually works for $3/4 < \eps \leq 1$.}

%% file: includes/algorithms.tex
\newcommand{\tree}{\ensuremath{T}\xspace}
\newcommand{\rootOf}[1]{\ensuremath{r}\xspace}
\newcommand{\pathError}[1]{\ensuremath{e(#1)}\xspace}
\newcommand{\errorRange}[2][]{\ensuremath{E#1(#2)}\xspace}
\newcommand{\errorRangeSet}[2][]{\ensuremath{\mathcal{E}#1(#2)}\xspace}
\DontPrintSemicolon
\setlength\algomargin{0em}
\SetAlCapHSkip{0em}
\SetFuncSty{textsc}
\SetProcNameSty{textsc}
\SetArgSty{textrm}
\SetKwInOut{Input}{input}\SetKwInOut{Output}{output}

\section{A quadratic-time algorithm for trees}
In this section, we will present algorithms for the special case in which the graph is a tree. Note that in this case, there is only one simple path between any pair of vertices, so there is no difference between path-oblivious, weak, and strong \roundings.

Clearly, if the whole graph is a simple path with edges $e_1,...,e_n$, a 1-rounding always exists, and can be computed in linear time (assuming the floor function can be computed in constant time). For example~\cite{SadakaneEtAl2005,Storandt2018}, let $d_i$ be $\frac12 + \sum_{j=1}^i \weights(e_i)$; then we set $\roundedWeights(e_i) = \lfloor d_i \rfloor - \lfloor d_{i-1} \rfloor$. Now, for any subpath $e_a,...,e_z$, we have
$\sum_{i=a}^z \roundedWeights(e_i) = \lfloor d_z \rfloor - \lfloor d_{a-1} \rfloor < d_z - (d_{a-1} - 1) = 1 + \sum_{i=a}^z \weights(e_i)$, 
and 
$\sum_{i=a}^z \roundedWeights(e_i) = \lfloor d_z \rfloor - \lfloor d_{a-1} \rfloor > (d_z - 1) - d_{a-1} = -1 + \sum_{i=a}^z \weights(e_i)$; thus \roundedWeights satisfies Condition~\ref{condition:weights} for $\eps = 1$, and \roundedWeights is a 1-rounding. Sadakane et al.~\cite{SadakaneEtAl2005} prove that a path of $n$ vertices admits at most $n$ different 1-roundings, and shows how to compute all 1-roundings in $O(n^2)$ time, and how to determine the 1-rounding with the smallest maximum absolute rounding error in the same time.

If the graph is a tree, observe that we can obtain a 2-rounding in linear time as follows. 
Choose any vertex of the tree as the root~$r$.
For any other vertex $u$, let $p(u)$ be the parent of $u$, and let $d_u$ be the (unrounded) weight of the path from $r$ to $u$.
Now we set $\roundedWeights(\{p(u),u\}) = \lfloor d_u \rfloor - \lfloor d_{p(u)}\rfloor$.
By the same calculation as above, for any vertex~$u$, the absolute rounding error $|e(u,v)|$ on any path from $u$ to an ancestor $v$ of~$u$ is now less than one.
Now, given two arbitrary vertices $u$ and $w$, let $v$ be their lowest common ancestor.
The absolute rounding error on the path from $u$ to $w$ is at most $|e(u,v)| + |e(v,w)| < 2$.

We will now present an algorithm that decides, given a tree $\tree$ and an error threshold $\eps < 2$, in quadratic time, whether $\tree$ admits an \rounding.
We choose an arbitrary vertex of $\tree$ as the root $\rootOf{\tree}$. 
We say $v$ is a descendant of $u$ if $u$ lies on the path from $\rootOf{\tree}$ to $v$.
For any vertex $u$, the subtree $\tree_u$ of $\tree$ is the subgraph of $\tree$ that is induced by all descendants of~$u$; this vertex $u$ is called the root of~$\tree_u$; see \autoref{figure:TreeEx}.
By $|\tree|$ we denote the number of vertices of $\tree$. By $\shortestPath{u,v}$ we denote the path in $\tree$ from $u$ to $v$. 

\begin{figure}[tb]
\centering
\includegraphics[scale=0.5]{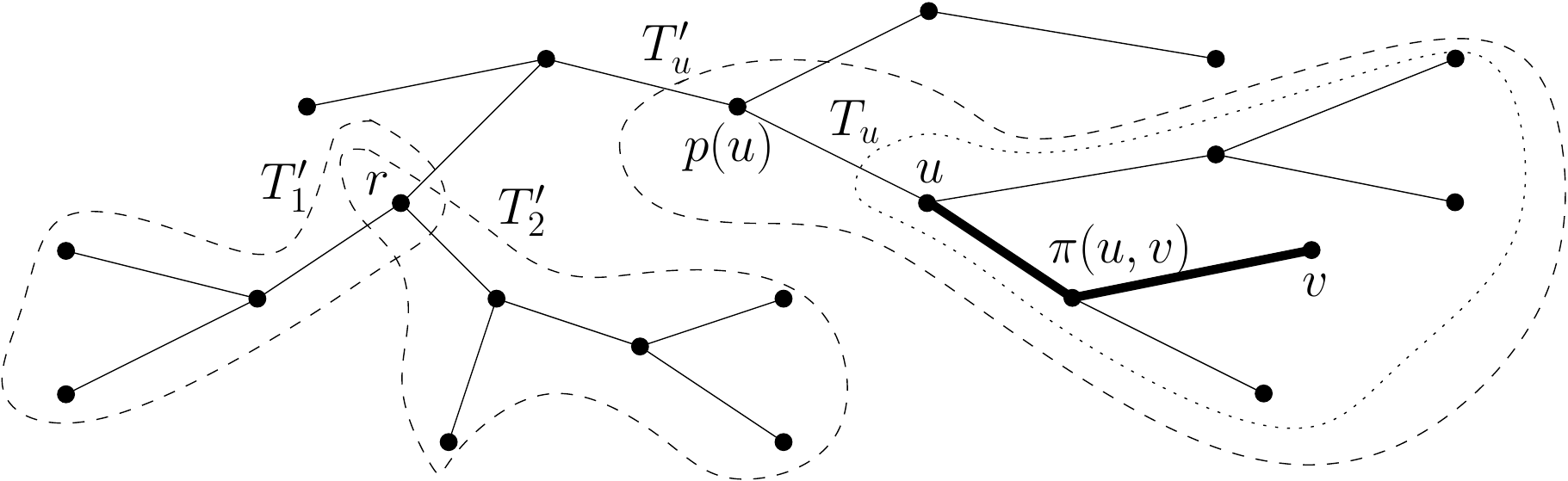}
\caption{A rooted tree $\tree$ with root vertex $r$. The subtree $\tree_u$ with root $u$ can be extended to a subgraph $T_u'$ with root $p(u)$ by adding the edge $\{u,p(u)\}$.
}\label{figure:TreeEx}
\end{figure}

\begin{definition}[root error range]
Let \roundedWeights be an \rounding on a tree \tree with root $\rootOf{\tree}$.
For any $v \in \tree$, let \pathError{r,v} be the rounding error on $\pi(r,v)$, that is, $\pathError{r,v} := \roundedWeights(\shortestPath{\rootOf{\tree},v}) - \weights(\shortestPath{\rootOf{\tree},v})$.

We call the smallest interval that contains the signed rounding errors of the paths from $r$ to all vertices of $\tree$ the \emph{root error range} \errorRange{\tree, \roundedWeights}, so
$$
\errorRange{\tree, \roundedWeights} := \left[ \min\limits_{v\in \tree} \pathError{\rootOf{\tree},v}, \max\limits_{v\in\tree} \pathError{\rootOf{\tree},v} \right] .
$$
Note that $\pathError{\rootOf{\tree},\rootOf{\tree}} = 0$, so if $\tree$ is a leaf, then $\errorRange{\tree, \roundedWeights} = [0, 0]$.

We call a rounding \roundedWeights of \tree \emph{locally optimal} if there is no other rounding $\roundedWeights'$ of \tree such that the corresponding root error range \errorRange{\tree, \roundedWeights'} is smaller than and contained in \errorRange{\tree, \roundedWeights}.

Let the \emph{error range set} \errorRangeSet{\tree} be the set of root error ranges that can be realized by locally optimal \roundings of \tree, that is, the set $\errorRangeSet{\tree} := \lbrace \errorRange{\tree, \roundedWeights} \mid \roundedWeights \text{ is a locally optimal rounding} \rbrace$. 
\end{definition}

\begin{lemma}[error range set size]
$\errorRangeSet{\tree}$ has at most~$2|T|$ elements.
\end{lemma}

\begin{proof}
Observe that in any \rounding with $\eps < 2$, the weight of any path $\shortestPath{\rootOf{\tree},v}$ is rounded to $\lfloor\weights(\shortestPath{\rootOf{\tree},v})\rfloor-1$, $\lfloor\weights(\shortestPath{\rootOf{\tree},v})\rfloor$, $\lceil\weights(\shortestPath{\rootOf{\tree},v})\rceil$, or $\lceil\weights(\shortestPath{\rootOf{\tree},v})\rceil+1$. 

Any root error range $\errorRange{\tree,\roundedWeights}$ includes 0, since $\pathError{\rootOf{\tree},\rootOf{\tree}} = 0$. Therefore, the lower bound of any root error range $\errorRange{\tree,\roundedWeights}$ is zero or negative, and it must be the rounding error on some path $\shortestPath{\rootOf{\tree},v}$ whose rounded weight is $\lfloor\weights(\shortestPath{\rootOf{\tree},v})\rfloor$ or $\lfloor\weights(\shortestPath{\rootOf{\tree},v})\rfloor - 1$. Since $v$ must be one of the $n$ vertices of $T$, this implies that there are at most $2n$ possible values for the lower bound of any root error range.

Because \errorRangeSet{\tree} contains only root error ranges of locally optimal \roundings, no two elements of \errorRangeSet{\tree} can have the same lower bound, so the total number of elements of \errorRangeSet{\tree} is also bounded by $2n$.
\end{proof}

Our algorithm will compute the error range set for every subtree of \tree bottom-up. For this purpose, we need two subalgorithms. The first subalgorithm (explained in the proof of Lemma~\ref{lem:addparentedge}) adds, to a given subtree, the edge that connects the root to its parent in $\tree$. The second subalgorithm (explained in the proof of Lemma~\ref{lem:merge}) combines two such parent-added subtrees who have a common parent. In the description of these algorithms, we assume that error range sets are sorted in ascending order by the lower bounds of the root error ranges. Since error range sets contain only root error ranges of locally optimal roundings, no element of an error range set can be contained in another. Therefore, the fact that the error range sets are sorted by ascending lower bound, implies that they are also sorted by ascending upper bound.

Let $\tree'_u$ be the subgraph of $\tree$ that consists of $\tree_u$ and the edge between $u$ and its parent $p(u)$ in $\tree$; we choose $p(u)$ as the root of $\tree'_u$, see \autoref{figure:TreeEx}.

\begin{lemma}[moving up]\label{lem:addparentedge}
Given \errorRangeSet{\tree_u}, we can compute \errorRangeSet{\tree'_u} in $O(|\tree_u|)$ time.
\end{lemma}
\begin{proof}
Let $f$ be the fractional part of the weight of $\{p(u),u\}$, that is, $f := \weights(\{p(u),u\}) - \lfloor\weights(\{p(u),u\})\rfloor$.
Any \rounding for $\tree'_u$ must consist of an \rounding for $\tree_u$ combined with setting $\roundedWeights(\{p(u),u\})$ to $\lfloor\weights(\{p(u),u)\}\rfloor + k$ for some $k \in \{-1,0,1,2\}$ (because $\eps < 2$, no other values for $\roundedWeights(\{p(u),u\})$ are allowed).
For any vertex $v \in \tree'_u$, other than the root $p(u)$, we have $\pathError{p(u),v} = \pathError{p(u),u} + \pathError{u,v} = k-f + \pathError{u,v}$; for the root $p(u)$ we have $\pathError{p(u),p(u)} = 0$. Thus, a choice of an \rounding for $T_u$ with root error range $[a,b] \in \errorRangeSet{\tree_u}$, together with a choice of $k \in \{-1,0,1,2\}$, results in a rounding for $T'_u$ whose root error range is the smallest interval that includes $\left[a+k-f, b+k-f\right]$ and~$0$, that is, the root error range for $T'_u$ is $\left[\min(a+k-f,0),\max(0,b+k-f)\right]$.
This rounding is an \rounding if and only if $-\eps < a+k-f$ and $b+k-f < \eps$.

Thus, the elements of $\errorRangeSet{\tree'_u}$ are all from the set:
\begin{eqnarray*}
S = & \left\lbrace 	\left[\min(a+k-f,0),\max(0,b+k-f)\right]
~\left|~
\begin{tabular}{c}
$[a,b] \in \errorRangeSet{\tree_u},$\\
$k \in \{-1,0,1,2\},$ \\
$-\eps < a+k-f,$ \\
$b+k-f < \eps$
\end{tabular}\right.\right\rbrace.
\end{eqnarray*}
We can compute $S$ in lexicographical order by first computing, for each $k \in \{-1,0,1,2\}$, the set $S_k := \{[\min(a+k-f,0),\max(0,b+k-f)] \mid [a,b] \in \errorRangeSet{\tree_u}, -\eps < a+k-f, b+k-f < \eps\}$ in lexicographical order from $\errorRangeSet{\tree_u}$, and then merging the sets $S_{-1}, S_0, S_1$ and $S_2$ into one lexicographically ordered set~$S$. 

To obtain $\errorRangeSet{\tree'_u}$ from $S$, all that remains to do is to filter out the root error ranges that are not locally optimal. To do so, we start with an empty stack and scan $S$ in lexicographical order. When we scan an element $[a,b]$ of $S$, we pop elements from the stack until the stack is empty or until the top element $[a',b']$ satisfies $b' < b$; then, if $a' \neq a$, we push $[a,b]$ onto the stack. After all elements of $S$ have been scanned, the stack contains $\errorRangeSet{\tree'_u}$. 

Procedure \FuncSty{Filter} describes the filtering algorithm in pseudocode; Procedure \FuncSty{ComputeErrorRangeSetWithParentEdge} gives the full algorithm to compute $\errorRangeSet{\tree'_u}$ from $\errorRangeSet{\tree_u}$.

\begin{procedure}
\caption{Filter(${\cal E}$)}
\label{alg:filter}
\Input{An set ${\cal E}$ of root error ranges in lexicographical order.}
\Output{Maximal subset ${\cal E}'$ of ${\cal E}$ such that no element of ${\cal E}'$ is contained in another.}
${\cal E}' \gets$ list with sentinel element $[-\infty,-\infty]$\;
\ForEach{$[a,b]$ in ${\cal E}$}{
  $[a',b'] \gets$ last element of ${\cal E}'$\;
  \lWhile{$b' \geq b$}{Remove last element from ${\cal E}'$; $[a',b'] \gets$ last element of ${\cal E}'$}
  \lIf{$a' \neq a$}{Append $[a,b]$ to ${\cal E}'$}
}
\Return ${\cal E}'$ without the first element (sentinel)
\end{procedure}

\begin{procedure}
\caption{ComputeErrorRangeSetWithParentEdge($\errorRangeSet{\tree_u}$)}
\label{alg:addparentedge}
\Input{$\errorRangeSet{\tree_u}$ in ascending order, where $\tree_u$ has root $u$.}
\Output{$\errorRangeSet{\tree'_u}$ in ascending order, where $\tree'_u = \tree_u \cup \{\{p(u),u\}\}$ with root $p(u)$.}
\BlankLine
$f \gets \weights(\{p(u),u\}) - \lfloor\weights(\{p(u),u\})\rfloor$\;
\For{$k \gets -1$ \KwTo $2$}{
  $S_k \gets$ empty list\;
   \ForEach{$[a,b]$ in $\errorRangeSet{\tree_u}$}{
    \If{$-\eps < a+k-f$ and $b+k-f<\eps$}{Append $[\min(a+k-f,0),\max(0,b+k-f)]$ to $S_k$}
  }
}
Merge $S_{-1},S_0,S_1$ and $S_2$ into lexicographically ordered list $S$\;
\Return $\FuncSty{Filter}(S)$\;
\end{procedure}

Clearly, the filtering algorithm runs in linear time: for each element $[a,b]$ of~$S$ we spend time proportional to the number of pushes (appends) and pops (removes), and each element is pushed at most once and popped at most once. The correctness follows from two observations. First, we maintain the invariant that the root error ranges on the stack, from bottom to top, are sorted in ascending order by lower bound \emph{and} by upper bound (so that none of these intervals contains another). This invariant is maintained because we only push $[a,b]$ onto the stack when the stack is empty, or when the top $[a',b']$ of the stack satisfies both $a' < a$ and $b' < b$. Second, we only discard root error ranges that are not locally optimal. To see this, observe that an element $[a',b']$ is removed from the stack only when we scan an element $[a,b]$ with $a' \leq a$ (because $[a,b]$ follows $[a',b']$ in the scanning order) and $b \leq b'$ (otherwise we would stop popping). This implies $[a,b] \in [a',b']$. If $[a,b] = [a',b']$, we will now pop $[a',b']$ but push $[a,b]$ onto the stack again and the stack does not actually change. Otherwise, $[a,b]$ is contained in and smaller than $[a',b']$, so $[a',b']$ does not correspond to a locally optimal rounding and is rightfully removed from the stack. We only refrain from pushing $[a,b]$ when the top of the stack satisfies $a = a'$ and $b' < b$, so that $[a',b']$ is contained in and smaller than $[a,b]$ and therefore, the root error range $[a,b]$ is not locally optimal.
\end{proof}

Given two trees $T'_1$ and $T'_2$ that have the same root vertex $r$, but are otherwise disjoint as in \autoref{figure:TreeEx}, we denote by $T'_1 \cup T'_2$ the union of the two trees; $T'_1 \cup T'_2$ also has root~$r$.

\begin{lemma}[merging error range sets]\label{lem:merge}
Given \errorRangeSet{\tree'_1} and \errorRangeSet{\tree'_2} for two trees $\tree'_1$ and $\tree'_2$, whose root $r$ is the only vertex that they have in common, we can compute \errorRangeSet{\tree'_1 \cup \tree'_2} in $O(|\tree'_1| + |\tree'_2|)$ time.
\end{lemma}
\begin{proof}
Consider a rounding $\roundedWeights$ of $\tree'_1 \cup \tree'_2$ that consists of an \rounding of $\tree'_1$ with root error range $[a_1,b_1]$ and an \rounding of $\tree'_2$ with root error range $[a_2,b_2]$. For $i \in \{1,2\}$, let $u_i$ and $v_i$ be vertices in $\tree'_i$ that determine the lower and upper bounds of the root error range $[a_i,b_i]$, that is: $a_i = \pathError{r,u_i}$ and $b_i = \pathError{r,v_i}$. The path composed of $\shortestPath{u_1,r}$ and $\shortestPath{r,u_2}$ is a path in $\tree'_1 \cup \tree'_2$ with $\pathError{u_1,u_2} = a_1 + a_2$; similarly, we have $\pathError{v_1,v_2} = b_1 + b_2$. It follows that $\roundedWeights$ can be an \rounding only if $a_1 + a_2 > -\eps$ and $b_1 + b_2 < \eps$. These conditions are also sufficient, since any other path from a vertex $w_1 \in \tree'_1$ to a vertex $w_2 \in \tree'_2$ consists of a path with error $\pathError{r,w_1} \in [a_1,b_1]$ and a path with error $\pathError{r,w_2} \in [a_2,b_2]$, so the total error is within $[a_1+a_2, b_1+b_2]$.

Since $\tree'_1$, $\tree'_2$ and $\tree'_1 \cup \tree'_2$ have the same root, the root error range \errorRange{\tree'_1 \cup \tree'_2, \roundedWeights} is the union of the root error ranges of $\tree'_1$ and $\tree'_2$, that is, $\errorRange{\tree'_1 \cup \tree'_2, \roundedWeights} = [\min(a_1,a_2), \max(b_1,b_2)]$. We say $\roundedWeights$ is of type 1 if $a_1 < a_2$, and of type 2 if $a_2 \leq a_1$. 

We will now explain how to find a linear-size set $S_1$ of root error ranges for $\tree'_1 \cup \tree'_2$ that includes the root error ranges of all locally optimal roundings of type 1. Recall that such a root error range for $\tree'_1 \cup \tree'_2$ must stem from a root error range $[a_1,b_1] \in \errorRangeSet{\tree'_1}$ and a root error range $[a_2,b_2] \in \errorRangeSet{\tree'_2}$ that satisfy the following conditions: (i) $a_2 > a_1$ (condition for type 1); (ii) $a_2 > -\eps-a_1$; (iii) $b_2 < \eps - b_1$. The idea of the algorithm is to scan the root error ranges $[a_1,b_1] \in \errorRangeSet{\tree'_1}$ in ascending order while maintaining pointers to: the first range $[a'_2,b'_2] \in \errorRangeSet{\tree'_2}$ that satisfies condition (i); the first range $[a''_2,b''_2] \in \errorRangeSet{\tree'_2}$ that satisfies condition (ii); and the last range $[a'''_2,b'''_2] \in \errorRangeSet{\tree'_2}$ that satisfies condition (iii). (Consider $\errorRangeSet{\tree'_2}$ augmented with sentinel ranges $[-\eps,-\eps]$ and $[\eps,\eps]$.) Note that as $a_1$ and $b_1$ increase, the first pointer ascends in $\errorRangeSet{\tree'_2}$, whereas the second and the third pointer descend in $\errorRangeSet{\tree'_2}$. Therefore we can scan all of $\errorRangeSet{\tree'_1}$ while maintaining the three pointers into $\errorRangeSet{\tree'_2}$ in linear time. 

For any error range $[a_1,b_1]$ scanned from $\errorRangeSet{\tree'_1}$, consider now the range $[a_2,b_2] \in \errorRangeSet{\tree'_2}$ with $a_2 = \max(a'_2, a''_2)$, that is, the range pointed to by the furthest of the first two pointers. If $b_2 \leq b'''_2$ (that is, we are not past the third pointer), we include $[\min(a_1,a_2), \max(b_1,b_2)] = [a_1, \max(b_1,b_2)]$ in $S_1$. Note that we do not need to consider a combination of $[a_1,b_1]$ with any other range $[a,b] \in \errorRangeSet{\tree'_2}$ between $[a_2,b_2]$ and $[a'''_2,b'''_2]$: the resulting error range for $\tree'_1 \cup \tree'_2$ would be $[a_1, \max(b_1,b)]$ and include $[a_1, \max(b_1,b_2)]$, so it would be either a duplicate of $[a_1, \max(b_1,b_2)]$, or it would not be locally optimal.

In a similar fashion, we can find a linear-size set $S_2$ of root error ranges for $\tree'_1 \cup \tree'_2$ that includes the root error ranges of all locally optimal roundings of type 2. Finally we can merge $S_1$ and $S_2$ and filter out error ranges that are not locally optimal with the algorithm described in the proof of Lemma~\ref{lem:addparentedge}. The full algorithm is given by Procedure \FuncSty{MergeErrorRangeSets}.
\end{proof}

\begin{procedure}
\caption{MergeErrorRangeSets($\errorRangeSet{\tree'_1}, \errorRangeSet{\tree'_1}$)}
\label{alg:merge}
\Input{$\errorRangeSet{\tree'_1}$ and $\errorRangeSet{\tree'_2}$, each in ascending order, where $\tree'_1$ and $\tree'_2$ have a common root.}
\Output{$\errorRangeSet{\tree'_1 \cup \tree'_2}$ in ascending order.}
\BlankLine
Add sentinel $[-\eps,-\eps]$ at the beginning of $\errorRangeSet{\tree'_1}$ and $\errorRangeSet{\tree'_2}$\;
Add sentinel $[\eps,\eps]$ at the end of $\errorRangeSet{\tree'_1}$ and $\errorRangeSet{\tree'_2}$\;
$S_1 \gets$ empty list\;
Let $p_1$ point to the first element of $\errorRangeSet{\tree'_2}$\;
Let $p_2$ and $p_3$ point to the last element of $\errorRangeSet{\tree'_2}$\;
\ForEach{$[a_1,b_1]$ in $\errorRangeSet{\tree'_1}$ except the sentinels}{
  \While{element $[a_2,b_2]$ pointed at by $p_1$ violates $a_2 > a_1$}{
    $p_1 \gets$ pointer to successor
  }
  \While{predecessor $[a_2,b_2]$ of el. pointed at by $p_2$ satisfies $a_2 > -\eps - a_1$}{
    $p_2 \gets$ pointer to predecessor
  }
  \While{element $[a_2,b_2]$ pointed at by $p_3$ violates $b_2 < \eps - b_1$}{
    $p_3 \gets$ pointer to predecessor
  }
  Let $[a'_2,b'_2], [a''_2, b''_2], [a'''_2,b'''_2]$ be the elements pointed at by $p_1, p_2, p_3$\;
  \lIf{$\max(b'_2, b''_2) \leq b'''_2$}{Append $[a_1,\max(b_1,b'_2,b''_2)]$ to $S_1$}  
}
Compute $S_2$ in a similar manner\;
Merge $S_1$ and $S_2$ into a lexicographically ordered list $S$\;
\Return $\FuncSty{Filter}(S)$\;
\end{procedure}

To decide whether a tree $\tree$ admits an \rounding, we compute $\errorRangeSet{\tree_u}$ for all subtrees $\tree_u$ of $\tree$ bottom-up. 

Specifically, if $u$ is a leaf, $\errorRangeSet{\tree_u} = [0,0]$. If $u$ is an internal vertex with a single child $v$, then $\tree_u = \tree'_v$ and we compute $\errorRangeSet{\tree_u} = \errorRangeSet{\tree'_v}$ from $\errorRangeSet{\tree_v}$ with the algorithm of Lemma~\ref{lem:addparentedge}. If $u$ is an internal vertex with two children $v$ and $w$, we first compute $\errorRangeSet{\tree'_v}$ and $\errorRangeSet{\tree'_w}$ from $\errorRangeSet{\tree_v}$ and $\errorRangeSet{\tree_w}$, respectively, with the algorithm of Lemma~\ref{lem:addparentedge}, and then we compute $\errorRangeSet{\tree_u} = \errorRangeSet{\tree'_v \cup \tree'_w}$ with the algorithm of Lemma~\ref{lem:merge}. Finally, if $u$ is an internal vertex with more than two children, we first compute $\errorRangeSet{\tree'_v}$ from $\errorRangeSet{\tree_v}$ for each child $v$. Then we organize all children in a balanced binary merge tree $M$ with the children of $u$ at the leaves \HH{here we could consider a figure too}; for a vertex $x$ in $M$, let $C(x)$ be the children of $u$ in the subtree of $M$ rooted at~$x$. With vertex $x$ we associate the error range set $\errorRangeSet{\bigcup_{v \in C(x)} \tree'_v}$. We process the merge tree $M$ bottom-up, using the algorithm of Lemma~\ref{lem:merge} for each internal vertex $x$ of $M$ to compute $\errorRangeSet{\bigcup_{v \in C(x)} \tree'_v}$ from the error range sets associated with the children of $x$. The error range set computed for the root of $M$ constitutes $\errorRangeSet{\tree_u}$.

Ultimately, we compute $\errorRangeSet{\tree_{\rootOf{\tree}}}$. If and only if this error range set is non-empty, $\tree$ admits an \rounding.

We say the \emph{effective} height of the tree $\tree$ is the height it would have when all internal vertices with more than two children were replaced by their binary merge trees. The algorithms of Lemmas \ref{lem:addparentedge} and \ref{lem:merge} take time linear in the size of the subtrees that are being processed. Thus, if $\tree$ has $n$ vertices and effective height $h$, the above algorithm to compute $\errorRangeSet{\tree_{\rootOf{\tree}}}$ runs in $O(nh)$ time. This proves:
\begin{theorem}
Given an edge-weighted tree \tree of $n$ vertices and an error tolerance~$\eps$, 
one can decide in $O(n^2)$ time whether \tree admits an \rounding.
\end{theorem}
To find the minimal maximum rounding error, we first compute the lengths of all $O(n^2)$ simple paths in the tree in $O(n^2)$ time. We can do so with a bottom-up algorithm that computes for each vertex $u$ the lengths of all paths in $\tree_u$, and passes on the lengths of all paths in $\tree_u$ that end in $u$ to the parent of $u$. Each path produces up to four candidate values for the maximum rounding error, namely, for $k \in \{-1,0,1,2\}$, the absolute value of ($k$ minus the fractional part of the path length). We sort all of these candidate values in $O(n^2 \log n)$ time. Finally we find the smallest error tolerance for which the decision algorithm says yes by binary search, using $O(\log n)$ calls to the decision algorithm, which takes $O(n^2 \log n)$ time in total. 

\begin{corollary}
Given an edge-weighted tree $\tree$ of $n$ vertices, we can compute a rounding of $\tree$ that minimizes the maximum absolute rounding error on any simple path in the tree in $O(n^2 \log n)$ time.
\end{corollary}


%% file: includes/conclusion.tex
\section{Conclusions and comparison to related work}

We have shown that it is, in general, NP-hard to decide whether a path-oblivious, weak, or strong \rounding exists for a given graph, but the problem can be solved in polynomial time if the graph is a tree.
Does this mean there is no hope of finding efficient algorithms to round weights in practical graphs other than trees? The conditions of our NP-hardness construction raise several questions.

First, note that we motivated the study of the \rounding problem with applications of graphs that are, typically, almost planar---not just trees, and not at all like the graph in our NP-hardness proof where most pairs of vertices are connected by a direct edge. To deal with realistic graphs in data structures for shortest-path queries, Storandt~\cite{Storandt2018} proposes to augment graphs with additional edges that can represent paths of many edges in the original graphs without suffering from accumulated rounding errors. However, for the map drawing application that we mentioned in the introduction, such an approach would not be suitable---we must really keep to rounding the weights of the original edges. Where then, between trees and almost-complete graphs, lies the boundary between easy and hard? So far, we have been unable to reduce \threeSAT, even planar \threeSAT, to a planar instance of an \rounding problem. It might be that in a (near-)planar graph, the dependencies between shortest paths between different pairs of vertices are so strong that the problem becomes easy to decide---as is the case with trees. However, we do not see how to adapt our algorithm for trees to planar graphs. Possibly, a first step in that direction would be to develop an efficient algorithm for graphs that consist of a single cycle, or trees attached to a single cycle, that is, a tree with one additional edge. To improve our understanding of the structure of the problem, we may also try to get a subquadratic algorithm for trees.

Second, we observe that the \NP-hardness construction as described works as long as $7/8 < \eps \leq 1$.
However, we do not know how to construct a working variable gadget for $\eps > 1$.
Does the problem remain NP-hard for values of $\eps$ (slightly) larger than 1, or does it become easy to solve in that case? \HH{I believe in the IWOCA version we also pose the question for $\eps < 7/8$. In fact, I believe the current proof can be adapted to $3/4 < \eps \leq 1$, and $\eps \leq 1/2$ is trivial: each edge can only be rounded in one way (if at all) and then you just need to check in polynomial time whether the solution is ok with respect to all shortest paths. So the open questions are really: $1/2 < \eps \leq 3/4$ and $\eps > 1$.}

Third, we might reconsider the exact conditions of an \rounding. Funke and Storandt~\cite{FunkeStorandt2016} study the rounding problem with different conditions. Whether this reduces the theoretical complexity of the problem as compared to the conditions in our paper is not clear. Funke and Storandt observe that many rounding problems are NP-hard, but do not prove this specifically for shortest-path-preserving rounding problems. They describe an ILP-formulation that they find to be too expensive to solve even for small graphs, and then describe and evaluate a greedy rounding heuristic. 


The essential difference in conditions is that Funke and Storandt consider relative errors rather than absolute errors: if $x$ is the distance between $u$ and $v$ before rounding, and $\tilde{x}$ is the distance between $u$ and $v$ after rounding, then the rounding error would be $\max(x/\tilde{x}, \tilde{x}/x)$ (where 1 means: no error). This might make the problem easier, because of the following propery of relative errors (which Funke and Storandt exploit in their heuristics): if the relative change of weight on any subpath of a path is bounded, then the relative change of weight on the complete path automatically adheres to the same bound. 

To prevent the rounding errors on short paths from dominating the result, Funke and Storandt introduce an input parameter $k$. They only require the relative rounding error to be at most $1 + \eps$ for distances between vertices that are at least $k$ edges apart along the shortest path. For individual edges, they require an absolute rounding error less than~1. Note that this implies that the rounding error that is allowed on a path from $u$ to $v$ depends on how the path is subdivided into edges: if the path consists of $m$ edges where $m < k$, an \emph{absolute} rounding error close to $m$ may be accepted; if $m > k$, the \emph{relative} rounding error is bounded to $1 + \eps$. To remove the dependency on $k$ and on additional vertices along a path, we propose the following alternative: for any shortest path~$\pi$ (regardless of the number of edges it consists of), Condition~\ref{condition:weights} of an \rounding becomes: $\min(\weights(\pi)/(1+\eps), \weights(\pi)-1) < \roundedWeights(\pi) < \max(\weights(\pi)\cdot(1+\eps),\weights(\pi)+1)$. In words: any shortest path should adhere to the relative error bound $1+\eps$ or to the absolute error bound~$1$. How would this affect our computational complexity bounds?

Finally, we note that various authors have studied roundings in the following setting. The input is a hypergraph $H$---to distinguish it from the graphs in our paper, we will call the vertices of $H$ \emph{hypervertices} and its edges \emph{hyperedges}. The hypervertices have real weights in $[0,1]$. Each hyperedge is a set of at least two hypervertices; its weight is the sum of the weights of its hypervertices. The goal is to find a \emph{global rounding}, that is, replace the hypervertex weights by integers such that the change of weight on each hypervertex and each hyperedge is less than one. Note that our path-oblivious \rounding[1] problem can be formulated in these terms. 

Asano et al.~\cite{AsanoEtAl2000a,AsanoEtAl2000b} proved that finding a global rounding for a hypergraph is NP-hard if the hypervertices represent cells of a square grid and the hyperedges represent squares of $2\times2$ cells. Later, Asano, with different co-authors~\cite{AsanoEtAl2004} studied the following case: the hypervertices represent the $n$ vertices of a graph $\graph$ with weights on edges and vertices, and the hyperedges represent all shortest paths in~\graph with respect to the edge weights. They conjectured that in this case, at most $n+1$ global roundings of the (hyper-)vertex weights are possible. This was proven for path-shaped graphs~\cite{SadakaneEtAl2005}, and later also for outerplanar graphs~$\graph$~\cite{TakkiChebihiTokuyama2003}. However, to establish a relation to our (path-oblivious, weak, or strong) \rounding problem, one would have to reduce the square grid rounding problem to our problem, or our problem of finding a rounding of the edge weights to the problem of finding a rounding of the (hyper-)vertex weights. Currently we do not see how to do this (except if $G$ is a simple path), but it might be worthwile to investigate this further.

%% file: includes/bibliography.tex
\small
\bibliographystyle{abbrv}